\documentclass[12pt]{article}

\usepackage{amsmath,graphicx,amsthm}

\newtheorem{theorem}{Theorem}[section]
\newtheorem{corollary}[theorem]{Corollary}

\newtheorem{proposition}[theorem]{Proposition}

\newcommand{\f}{\operatorname}

\title{A Modified Reference Prior for the Generalized Gamma Distribution}
\author{Pedro L. Ramos\thanks{ Email: pedrolramos@hotmail.com}  
\ \ and \ \ Francisco Louzada\thanks{ Email: louzada@icmc.usp.br} \\ \\ Institute of Mathematical Science and Computing \\ Universidade de Sao Paulo, Sao Carlos-SP, Brazil } 
\date{December 15, 2014}

\begin{document}
\maketitle

\begin{abstract}
In this paper we propose an objective Bayesian estimation approach for the parameters of the generalized gamma distribution. Various reference priors are obtained, but showing that they lead to improper posterior distributions. 
We overcome this problem by proposing a modification in a reference priori distribution, allowing for a proper posterior distribution for the parameters of the generalized gamma distribution.
We perform a simulation study in order to study the efficiency of the proposed methodology, which is also fully illustrated on a real data set.

\end{abstract}

\section{Introduction}

Introduced by Stacy (1962), the Generalized Gamma distribution (GG) is an important distribution that has proven to be very flexible in practice for modeling data from several areas, such as climatology, meteorology medicine, reliability and image processing data, among others. The GG distribution is a general distribution which has several particular cases, such as the exponential, Weibull, gamma, log-normal,  generalized normal, half-normal, Rayleigh, Maxwell-Boltzmann and chi distributions.
 Marani et al. (1986) use this distribution to analyze data relating to air quality in Venice, Italy. Thai and Meyer (1999) propose new methods for analyzing citations in recent publications to find journals with greater influence using the GG distribution. Aalo, et al. (2005) successfully used this distribution to analyze the performance degradation of wireless communication systems. Li et al. (2011) use the GG  distribution to obtain different techniques for processing SAR (Synthetic aperture radar) images. Other applications of the GG distribution can be seen in Noortwijk (2001), Dadpay et al. (2007), Balakrishnan e Peng (2006),  Raju e Srinivasan, (2002), e Ahsanullah et al. (2013).	

Using the classical approach, Hargar \& Bain (1970) show that the non-linear equations obtained by Stacy \& Mihram (1965) are very unstable. DiCiccio (1987) obtained approximate inferences based on approximations of normal distributions, however, as an approximated method, the obtained inference is inaccurate for small sample sizes. Huang \& Hwang (2006) use the method of moments to perform inferences for the GG distribution. Khodabin \& Ahmadabadi (2010) compare the method of moments with the maximum likelihood method and conclude that, in general, the maximum likelihood estimators (MLE's) have superior performance. However, the MLE values may heavily depend on the initial values for the parameters, needed for initiate iterative methods.

Using the Bayesian approach, Chang \& Kim (2011) deals with non-informative prior for the GG distribution.  
However, the authors consider that one of the parameters of the distribution is known.
Maswadah et. al (2013) also consider the case where one of the parameters is known and perform inference for the GG distribution based on order statistics. 
At least in principle, they methods are not recommended for real applications since, usually, all the parameters are unknown needing to be estimated. 

An important reference prior was introduced by Bernardo (1979), with further developments by Berger \& Bernardo (1989, 1992a, 1992b, 1992c). The proposed idea is to maximize the expected Kullback-Leibler divergence between the posterior distribution and the prior. The reference prior provides posterior distribution with interesting properties, such as invariance, consistent marginalization and consistent sampling properties (Bernardo, 2005). Some recent reference priors were obtained for the Pareto (Fu et. al, 2012),  Poisson-Exponential (Tomazella et. al, 2013), Extended Exponential Geometric (Ramos et. al, 2014), Inverse-Weibull (Kim et. al, 2014), Generalized Half-normal (Kang et. al, 2014). In this paper however we demonstrated that different reference priors for the GG distribution lead to improper posterior distributions. 

We overcome the problems presented above by proposing a modification in a reference prior distribution obtained when all the parameters are of interesting, leading to a proper posterior density. These results are of great practical interest since they finally enable the use of the GG distribution in practice from the Bayesian point of view, in a methodologically correct way, breaking with the problem of estimating the parameters of this important distribution.

The paper is organized as follows. In Section 2, we present the GG distribution and discuss its properties. In Section 3, we carry out a reference Bayesian analysis for this model. In this section we also present our modified reference prior. In Section 4 a simulation study is presented in order to study the efficiency of the proposed method. In Section 5 the methodology is illustrated in a real data set. Some final comments are presented in Section 6.

\section{Generalized Gamma Distribution}
\vspace{0.3cm}

A random variable T has a GG distribution if its density function is given by (Stacy, 1962),
\begin{equation}\label{denspgg}
f(t|\boldsymbol{\theta})= \frac{\alpha}{\Gamma(\phi)}\mu^{\alpha\phi}t^{\alpha\phi-1}\exp\left(-(\mu t)^{\alpha}\right) ,
\end{equation}
where $t>0$, $\boldsymbol{\theta}=(\phi,\mu,\alpha)$ and $\Gamma(\phi)=\int_{0}^{\infty}{e^{-x}x^{\phi-1}dx}$ is called gamma function. The shape parameters are given by $\alpha>0$ and $\phi >0$ and $\mu >0$ is a scale parameter. 

The cumulative distribution function is given by
\begin{equation}\label{densagg}
F(t|\boldsymbol{\theta})= \int_{0}^{(\mu t)^{\alpha}}{\frac{1}{\Gamma(\phi)}w^{\phi -1}e^{-w}}dw =\frac{\gamma\left[\phi,(\mu t)^{\alpha}\right]}{\Gamma(\phi)}\ ,
\end{equation}
where $\gamma(y,x)=\int_{0}^{x}{w^{y-1}e^{-w}}dw$ is called lower incomplete gamma.

Relevant probability distributions can be obtained from the GG distribution as the Weibull distribution ($\phi=1$) , Gamma distribution ($\alpha=1$), Log-Normal distribution (limit case when $\phi \rightarrow \infty$) and the Generalized Normal distribution ($\alpha=2$). The Generalized Normal distribution is also a distribution that includes several known distributions like, half-normal ($\phi=1/2,\mu=1/\sqrt{2}\sigma$), Rayleigh ($\phi=1,\mu=1/\sqrt{2}\sigma$), Maxwell-Boltzmann distribution ($\phi=3/2$) and chi distribution ($\phi=k/2 , k=1,2,\ldots$).

The mean and the variance of the GG distribution, are  respectively given by
\begin{equation}\label{mediagg}
E(T)=\frac{\Gamma\left(\phi + \frac{1}{\alpha}\right)}{\mu\Gamma(\phi)},
\end{equation}
\begin{equation}\label{vargg}
V(T)=\frac{1}{\mu^2}\left\{\frac{\Gamma\left(\phi + \frac{2}{\alpha}\right)}{\Gamma(\phi)}-\left(\frac{\Gamma\left(\phi + \frac{1}{\alpha}\right)}{\Gamma(\phi)}\right)^2\right\}.
\end{equation}

The Fisher information matrix given by,
\begin{equation}\label{mfishergg}
I(\alpha,\mu,\phi)=
\begin{bmatrix}
 \dfrac{1+2\psi(\phi)+\phi\psi ' (\phi)+\phi\psi(\phi)^2}{\alpha^2} & -\dfrac{1+\phi\psi(\phi)}{\mu} & -\dfrac{\psi(\phi)}{\alpha} \\
 -\dfrac{1+\phi\psi(\phi)}{\mu} & \dfrac{\phi\alpha^2}{\mu^2}  & \dfrac{\alpha}{\mu} \\
 -\dfrac{\psi(\phi)}{\alpha} & \dfrac{\alpha}{\mu} & \psi ' (\phi)
\end{bmatrix} .
\end{equation}

The hazard function is given by
\begin{equation}\label{fusobgg}
h(t|\boldsymbol{\theta})=\frac{f(t|\boldsymbol{\theta})}{S(t|\boldsymbol{\theta})}=\frac{\alpha\mu^{\alpha\phi}t^{\alpha\phi-1}\exp\left(-(\mu t)^{\alpha}\right)}{\Gamma\left[\phi,(\mu t)^{\alpha}\right]} ,
\end{equation}
where $\Gamma(y,x)=\int_{x}^{\infty}{w^{y-1}e^{-w}}dw$ is called upper incomplete gamma.

This model allows to obtain different forms of the hazard function, such as constant, increasing, decreasing, bathtub and unimodal shape. These properties make the GG distribution a  flexible model for reliability data. 
Figure 1 shows examples of the shapes of the probability density function and hazard function for different values of $\phi, \mu$ e $\alpha$.
\begin{figure}[!htb]
\centering
\includegraphics[scale=0.6]{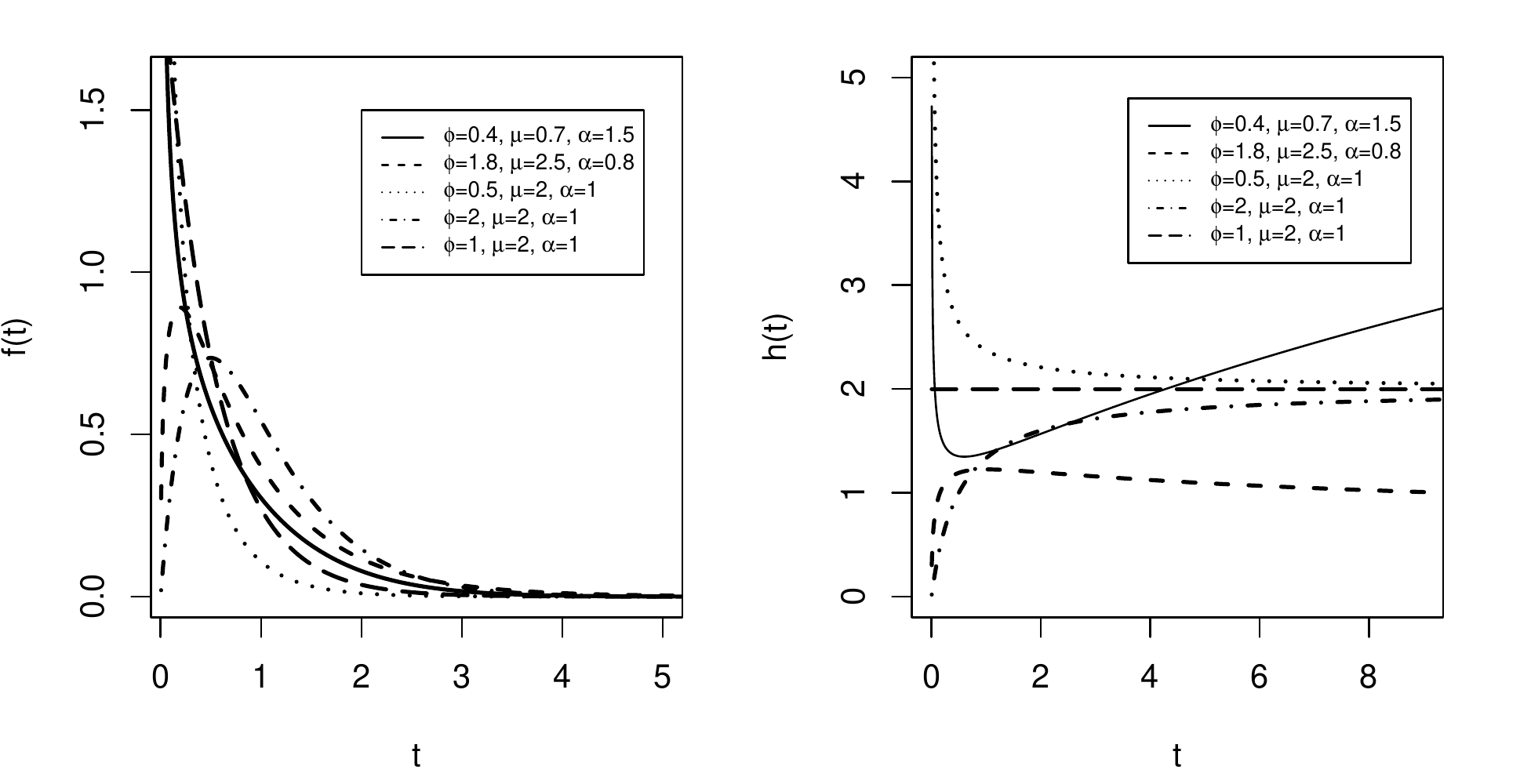}
\caption{Left panel: probability density function of the GG distribution. Right panel: hazard function of the GG distribution.}\label{fsobggen}
\end{figure}

\newpage
\section{Reference Bayesian Analysis}
\vspace{0.3cm}

The main objective of reference Bayesian analysis introduced by Bernardo (1979) and further developed by Berger and Bernardo (1992a, 1993b) is to specify a prior distribution, where the dominant information is provided by the data. To obtain the reference prior, when there are nuisance parameters, the construction must be performed using the parameters of interest ordered. The procedure for obtaining a reference prior to the parameter of interest is described as follows (Bernardo, 2005).

\begin{proposition} Let $p(x|\theta,\lambda_1,\cdots,\lambda_m)$, be a parametric model, the quantity of interest be $\theta$ and $H(\theta,\lambda_1,\cdots,\lambda_m)$ Fisher's information matrix $(m+1)\times(m+1)$. It is assumed that the joint distribution of $(\theta,\lambda_1,\cdots,\lambda_m)$ is asymptotically normal with the mean $(\hat\theta,\hat\lambda_1,\cdots,\hat\lambda_m)$ the estimating correspondents of maximum likelihood, and with the covariance matrix S$(\hat\theta,\hat\lambda_1,\cdots,\hat\lambda_m)$, where $S=H^{-1}$. If $S_j$ is a $j\times j$ upper left submatrix of $S$, $H_j= S_j^{-1}$ and $h_{i,j}(\theta,\lambda)$ is an element of $H_j$, it follows that:

\begin{enumerate}

\item The reference priors are
\begin{equation*}
\pi(\lambda_m|\theta,\lambda_1,\cdots,\lambda_{m-1})\propto h_{m+1,m+1}^{\frac{1}{2}}(\theta,\lambda_1,\cdots,\lambda_{m}), 
\end{equation*}
\begin{equation*}
\begin{aligned}
\pi(\lambda_i|\theta,\lambda_1,\cdots,\lambda_{i-1})\propto &\exp(\int_{\wedge_{i+1}}\cdots\int_{\wedge_{m}}\log\left(h_{i+1,i+1}^{\frac{1}{2}}(\theta,\lambda_1,\cdots,\lambda_{m})\right)\times\\ &\times\prod_{j=i+1}^{m}\pi(\lambda_j|\theta,\lambda_1,\cdots,\lambda_{j-1})\boldsymbol{d\lambda_{i+1}}), 
\end{aligned}
\end{equation*}
where $\boldsymbol{d\lambda_{i+1}}=d\lambda_j\times\cdots\times d\lambda_j$ and $\pi(\lambda_i|\theta,\lambda_1,\cdots,\lambda_{i-1}), i=1,\cdots,m$ are all proper. If any of those conditional reference priors is not proper, then a compact approximation is required for the corresponding integrals.

\item The marginal reference prior of $\theta$ is given by solving
\begin{equation*}
\begin{aligned}
\pi(\theta)\propto &\exp(\int_{\wedge_{i+1}}\cdots\int_{\wedge_{m}}\log\left(s_{1,1}^{-\frac{1}{2}}(\theta,\lambda_1,\cdots,\lambda_{m})\right)\times\\ &\times\prod_{j=i+1}^{m}\pi(\lambda_j|\theta,\lambda_1,\cdots,\lambda_{j-1})\boldsymbol{d\lambda_{i+1}}), 
\end{aligned}
\end{equation*}
where $s_{1,1}^{-\frac{1}{2}}(\theta,\lambda_1,\cdots,\lambda_{m})=h_{1,1}^{\frac{1}{2}}(\theta,\lambda_1,\cdots,\lambda_{m})$

\item The posterior distribution of $\theta$, given by the values $(y_1,\cdots,y_n)$ is
\begin{equation*}
\begin{aligned}
\pi(\theta|x_1,\cdots,x_n)\propto &\exp(\int_{\wedge_{i+1}}\cdots\int_{\wedge_{m}}p(x_1,\cdots,x_n|\theta,\lambda_1,\cdots,\lambda_{m})\times\\ &\times\prod_{j=i+1}^{m}\pi(\lambda_j|\theta,\lambda_1,\cdots,\lambda_{j-1})d\lambda_{1},\cdots,d\lambda_{m}), 
\end{aligned}
\end{equation*}
\end{enumerate}
\end{proposition}

\begin{corollary} If the nuisance parameter spaces $\wedge_{i}(\theta,\lambda_1,\cdots,\lambda_{j-1})=\wedge_{i}$ are independent of $\theta$ and $\lambda_i$'s and the functions $h_{i,i},\cdots,h_{m,m}$ factorize in the form
\begin{equation*}
s_{1,1}^{\frac{1}{2}}(\theta,\lambda_1,\cdots,\lambda_{m})=f_0(\theta)g_0(\lambda_1,\cdots,\lambda_{m}),
\end{equation*}
\begin{equation*}
h_{i+1,i+1}^{\frac{1}{2}}(\theta,\lambda_1,\cdots,\lambda_{m})=f_i(\lambda_i)g_i(\theta,\lambda_1,\cdots,\lambda_{i-1},\lambda_{i+1},\cdots,\lambda_{m}), i=1,\cdots,m.
\end{equation*}

Then
\begin{equation}
\pi(\theta)\propto f_0(\theta), \pi(\lambda_i|\theta,\lambda_1,\cdots,\lambda_{m})\propto f_i(\lambda_i), i=1,\cdots,m,
\end{equation}
and there is no need for compact approximations, even if  $\pi(\lambda_i|\theta,\lambda_1,\cdots,\lambda_{m})$ are not proper.
\end{corollary} 
\begin{corollary} 
If the parameter spaces $\wedge_{j}(\theta_1,\cdots,\theta_{j-1},\theta_{j+1},\cdots,\theta_{m})=\wedge_{j}$ are independent of $\theta_1,\cdots,\theta_{j-1},\theta_{j+1},\cdots,\theta_{m}$ and $h_{j,j}, j=1,\cdots,m,$ are factorize in the form
\begin{equation*}
h_{j,j}^{\frac{1}{2}}(\boldsymbol\theta)=f_j(\theta_j)g_j(\theta_1,\cdots,\theta_{j-1},\theta_{j+1},\cdots,\theta_{m}), j=1,\cdots,m.
\end{equation*}

Then the reference prior for the ordered parameters $\boldsymbol\theta=(\theta_1,\cdots,\theta_{m})$ is given by $\pi_{\boldsymbol\theta}(\boldsymbol\theta)=\prod_{i=1}^{m}f_j(\theta_j)$ and there is no need for compact approximations, even if the conditional priors are not proper.
\end{corollary}

\subsection{Reference prior when $\alpha$ is the parameter of interest}\label{sepalfa}
\vspace{0.3cm}

Suppose that $\alpha$  is the parameter of interest and $\mu$ and $\phi$ are nuisance parameters, through the Corollary 3.2 and using the Fisher information matrix (\ref{mfishergg}), we have
\begin{equation}\label{prioralfa1}
\pi(\alpha)\propto \alpha, \ \ \ \  \pi(\mu|\alpha,\phi)\propto\frac{1}{\mu}, \ \ \ \ \pi(\phi|\alpha,\mu)\propto\sqrt{\psi'(\phi)}.
\end{equation}

The joint prior distribution for the ordered parameters is given by,
\begin{equation}\label{prioralfa}
\pi_\alpha(\alpha,\mu,\phi)\propto \frac{\alpha\sqrt{\psi'(\phi)}}{\mu}.
\end{equation}

Let $T_1,\ldots,T_n$ be a random sample such that $T\sim \f {GG}(\alpha,\mu,\phi)$. In this case, the likelihood function from (\ref{denspgg}) is given by,
\begin{equation}\label{verogg1} 
L(\boldsymbol{\theta};\boldsymbol{t})=\frac{\alpha^n}{\f \Gamma(\phi)^n}\mu^{n\alpha\phi}\left\{\prod_{i=1}^n{t_i^{\alpha\phi-1}}\right\}\exp\left\{-\mu^{\alpha}\sum_{i=1}^n t_i^\alpha\right\}. \end{equation}

The joint posterior distribution for $\phi, \mu$ e $\alpha$, using the prior distribution (\ref{prioralfa}), is proportional to the  prior (\ref{prioralfa}) and the product of the likelihood function (\ref{verogg1}) resulting in,
\begin{equation}\label{postalfa1}
p_\alpha(\phi,\mu,\alpha|\boldsymbol{t})=\frac{1}{d_1}\frac{\alpha^{n+1}\sqrt{\psi'(\phi)}}{\Gamma(\phi)^n}\mu^{n\alpha\phi-1}\left\{\prod_{i=1}^n{t_i^{\alpha\phi-1}}\right\}\exp\left\{-\mu^{\alpha}\sum_{i=1}^n t_i^\alpha\right\}. 
\end{equation}

\begin{proposition} The posterior density (\ref{postalfa1}) is improper since,
\begin{equation}\label{postalfa2}
d_1=\int\limits_{\mathcal{A}}\frac{\alpha^{n+1}\sqrt{\psi'(\phi)}}{\Gamma(\phi)^n}\mu^{n\alpha\phi-1}\left\{\prod_{i=1}^n{t_i^{\alpha\phi-1}}\right\}\exp\left\{-\mu^{\alpha}\sum_{i=1}^n t_i^\alpha\right\}d\boldsymbol{\theta}=\infty,
\end{equation}
where $\mathcal{A}=\{(0,\infty)\times(0,\infty)\times(0,\infty)\}$ is the parameter space for $\boldsymbol{\theta}$.
\end{proposition}
\begin{proof}
Since $\frac{\alpha^{n+1}\sqrt{\psi'(\phi)}}{\Gamma(\phi)^n}\mu^{n\alpha\phi-1}\left\{\prod_{i=1}^n{t_i^{\alpha\phi-1}}\right\}\exp\left\{-\mu^{\alpha}\sum_{i=1}^n t_i^\alpha\right\}\geq 0$ by Tonelli theorem (see Folland, 1999) we have,
\begin{equation*}
\begin{aligned}
d_1&= \int\limits_0^{\infty}\int\limits_0^{\infty}\int\limits_0^{\infty}\frac{\alpha^{n+1}\sqrt{\psi'(\phi)}}{\Gamma(\phi)^n}\mu^{n\alpha\phi-1}\left\{\prod_{i=1}^n{t_i^{\alpha\phi-1}}\right\}\exp\left\{-\mu^{\alpha}\sum_{i=1}^n t_i^\alpha\right\}d\mu d\phi d\alpha \\
&=\int\limits_0^{\infty}\int\limits_0^{\infty}\alpha^{n-1}\frac{\sqrt{\psi'(\phi)}\Gamma(n\phi)}{\Gamma(\phi)^n}\frac{{\left(\prod_{i=1}^n t_i\right)}^{\alpha\phi-1}}{\left(\sum_{i=1}^n t_i^\alpha\right)^{n\phi}}d\phi d\alpha\\
& \geq \int\limits_0^{\infty}\int\limits_0^{1} c_1\alpha^{n}\phi^{n-2}\left(\frac{{\sqrt[n]{\prod_{i=1}^n t_i^\alpha}}}{\sum_{i=1}^n t_i^\alpha}\right)^{n\phi}d\phi d\alpha \\
&= \int\limits_0^{\infty}c_1\alpha^{n}\frac{\gamma\left(n-1,n \f{q}(\alpha)\right)}{\left(n \f{q}(\alpha)\right)^{n-1}}d\alpha \geq \int\limits_1^{\infty}g_1\alpha d\alpha = \infty, 
\end{aligned}
\end{equation*}
where $\f{q}(\alpha)=\log\left(\dfrac{\sum_{i=1}^n t_i^\alpha}{{\sqrt[n]{\prod_{i=1}^n t_i^\alpha}}}\right)>0$ and $c_1$ and $g_1$ are positive constants such that the above inequalities occur. \qedhere
\end{proof}

\subsection{Reference prior when $\phi$ is the parameter of interest}
\vspace{0.3cm}

When $\phi$  is the parameter of interest and $\mu$ and $\alpha$ are nuisance parameters, through the Corollary 3.2 and using the Fisher information matrix (\ref{mfishergg}), we have
\begin{equation}\label{priorphi1}
\pi(\phi)\propto \sqrt{\frac{\phi+\phi^2\psi'(\phi)-1}{\phi^2\psi'(\phi)^2-\psi'(\phi)-1}}, \ \ \ \  \pi(\mu|\phi,\alpha)\propto\frac{1}{\mu}, \ \ \ \ \pi(\alpha|\phi,\mu)\propto\frac{1}{\alpha}.
\end{equation}

The joint prior distribution for the ordered parameters is given by,
\begin{equation}\label{priorphi}
\pi_\phi(\phi,\alpha,\mu)\propto \frac{\pi(\phi)}{\alpha\mu}.
\end{equation}

The joint posterior distribution for $\phi, \mu$ e $\alpha$, using the prior distribution (\ref{priorphi}), is proportional to the product of the likelihood function (\ref{verogg1}) and the prior (\ref{prioralfa}) resulting in,
\begin{equation}\label{postphi1}
p_\phi(\phi,\mu,\alpha|\boldsymbol{t})=\frac{\pi(\phi)}{d_2}\frac{\alpha^{n-1}}{\Gamma(\phi)^n}\mu^{n\alpha\phi-1}\left\{\prod_{i=1}^n{t_i^{\alpha\phi-1}}\right\}\exp\left\{-\mu^{\alpha}\sum_{i=1}^n t_i^\alpha\right\}. 
\end{equation}

\begin{proposition}  The posterior density (\ref{postphi1}) is improper since,
\begin{equation}\label{postphi2}
d_2=\int\limits_{\mathcal{A}}\alpha^{n-1}\frac{\pi(\phi)}{\Gamma(\phi)^n}\mu^{n\alpha\phi-1}\left\{\prod_{i=1}^n{t_i^{\alpha\phi-1}}\right\}\exp\left\{-\mu^{\alpha}\sum_{i=1}^n t_i^\alpha\right\}d\boldsymbol{\theta}=\infty,
\end{equation}
where $\mathcal{A}=\{(0,\infty)\times(0,\infty)\times(0,\infty)\}$ is the parameter space for $\boldsymbol{\theta}$.
\end{proposition}
\begin{proof}
Since $\alpha^{n-1}\frac{\pi(\phi)}{\Gamma(\phi)^n}\mu^{n\alpha\phi-1}\left\{\prod_{i=1}^n{t_i^{\alpha\phi-1}}\right\}\exp\left\{-\mu^{\alpha}\sum_{i=1}^n t_i^\alpha\right\}\geq 0$ by Tonelli theorem (see Folland, 1999) we have,
\begin{equation*}
\begin{aligned}
d_2&= \int\limits_0^{\infty}\int\limits_0^{\infty}\int\limits_0^{\infty}\alpha^{n-1}\frac{\pi(\phi)}{\Gamma(\phi)^n}\mu^{n\alpha\phi-1}\left\{\prod_{i=1}^n{t_i^{\alpha\phi-1}}\right\}\exp\left\{-\mu^{\alpha}\sum_{i=1}^n t_i^\alpha\right\}d\mu d\phi d\alpha \\
&=\int\limits_0^{\infty}\int\limits_0^{\infty}\alpha^{n-2}\frac{\pi(\phi)\Gamma(n\phi)}{\phi\Gamma(\phi)^n}\frac{{\left(\prod_{i=1}^n t_i\right)}^{\alpha\phi-1}}{\left(\sum_{i=1}^n t_i^\alpha\right)^{n\phi}}d\phi d\alpha\\
& > \int\limits_0^{\infty}\int\limits_0^{1} c_2\alpha^{n-2}\phi^{\frac{n}{2}+1}\left(\frac{{\sqrt[n]{\prod_{i=1}^n t_i^\alpha}}}{\sum_{i=1}^n t_i^\alpha}\right)^{n\phi}d\phi d\alpha \\
&= \int\limits_0^{\infty}c_2\alpha^{n-2}\frac{\Gamma\left(\frac{n}{2}+2,n \f{p}(\alpha)\right)}{\left(n \f{p}(\alpha)\right)^{\frac{n}{2}+2}}d\alpha \geq \int\limits_1^{\infty}g_2\alpha^{-6}d\alpha = \infty, 
\end{aligned}
\end{equation*}
where $\f{p}(\alpha)=\log\left(\dfrac{\frac{1}{n}\sum_{i=1}^n t_i^\alpha}{{\sqrt[n]{\prod_{i=1}^n t_i^\alpha}}}\right)>0$ and $c_2$ and $g_2$ are positive constants such that the above inequalities occur. \qedhere
\end{proof}


\subsection{Reference prior when $\mu$ is the parameter of interest}
\vspace{0.3cm}

Consider the case when $\mu$ is the parameter of interest and $\phi$ and $\alpha$ are nuisance parameters, through the Corollary 3.2, we have
\begin{equation}\label{priormi1}
\pi(\mu)\propto \mu, \ \ \ \  \pi(\alpha|\mu,\phi)\propto\frac{1}{\alpha}, \ \ \ \ \pi(\phi|\alpha,\mu)\propto\sqrt{\psi'(\phi)}.
\end{equation}

The joint prior distribution for the ordered parameters is given by,
\begin{equation}\label{priormi}
\pi_\mu(\mu,\alpha,\phi)\propto \frac{\mu\sqrt{\psi'(\phi)}}{\alpha}.
\end{equation}

The joint posterior distribution for $\mu, \alpha$ e $\phi$, using the prior distribution (\ref{priormi}), is given by
\begin{equation}\label{postmu1}
p_\mu(\mu,\alpha,\phi|\boldsymbol{t})=\frac{1}{d_3}\frac{\alpha^{n-1}\sqrt{\psi'(\phi)}}{\Gamma(\phi)^n}\mu^{n\alpha\phi+1}\left\{\prod_{i=1}^n{t_i^{\alpha\phi-1}}\right\}\exp\left\{-\mu^{\alpha}\sum_{i=1}^n t_i^\alpha\right\}. 
\end{equation}

\begin{proposition}  The posterior density (\ref{postmu1}) is improper since,
\begin{equation}\label{postmu2}
d_3=\int\limits_{\mathcal{A}}\frac{\alpha^{n-1}\sqrt{\psi'(\phi)}}{\Gamma(\phi)^n}\mu^{n\alpha\phi+1}\left\{\prod_{i=1}^n{t_i^{\alpha\phi-1}}\right\}\exp\left\{-\mu^{\alpha}\sum_{i=1}^n t_i^\alpha\right\}d\boldsymbol{\theta}=\infty,
\end{equation}
where $\mathcal{A}=\{(0,\infty)\times(0,\infty)\times(0,\infty)\}$ is the parameter space for $\boldsymbol{\theta}$.
\end{proposition}
\begin{proof}
Since $\frac{\alpha^{n-1}\sqrt{\psi'(\phi)}}{\Gamma(\phi)^n}\mu^{n\alpha\phi+1}\left\{\prod_{i=1}^n{t_i^{\alpha\phi-1}}\right\}\exp\left\{-\mu^{\alpha}\sum_{i=1}^n t_i^\alpha\right\}\geq 0$ by Tonelli theorem (see Folland, 1999) we have,
\begin{equation*}
\begin{aligned}
d_3&= \int\limits_0^{\infty}\int\limits_0^{\infty}\int\limits_0^{\infty}\frac{\alpha^{n-1}\sqrt{\psi'(\phi)}}{\Gamma(\phi)^n}\mu^{n\alpha\phi+1}\left\{\prod_{i=1}^n{t_i^{\alpha\phi-1}}\right\}\exp\left\{-\mu^{\alpha}\sum_{i=1}^n t_i^\alpha\right\}d\mu d\phi d\alpha \\
&=\int\limits_0^{\infty}\int\limits_0^{\infty}\alpha^{n-2}\frac{\sqrt{\psi'(\phi)}\Gamma(n\phi+\frac{2}{\alpha})}{\Gamma(\phi)^n}\frac{{\left(\prod_{i=1}^n t_i\right)}^{\alpha\phi-1}}{\left(\sum_{i=1}^n t_i^\alpha\right)^{n\phi+\frac{2}{\alpha}}}d\phi d\alpha\\
& > \int\limits_0^{\infty}\int\limits_1^{\infty} c_3\alpha^{n-2}\phi^{\frac{n-1}{2}}\left(1+\frac{2}{n\alpha\mu}\right)^{n\phi}\left(\phi+\frac{2}{n\alpha}\right)^{\frac{2}{\alpha}-\frac{1}{2}}e^{-\frac{2}{\alpha}}\frac{{\left(\prod_{i=1}^n t_i\right)}^{\alpha\phi-1}}{\left(\sum_{i=1}^n t_i^\alpha\right)^{n\phi}}  d\phi d\alpha \\
& > \int\limits_0^{\infty}\int\limits_1^{\infty} c_3\alpha^{n-\frac{3}{2}}\phi^{\frac{n-1}{2}}\left(\frac{2}{n\alpha}\right)^{\frac{2}{\alpha}}e^{-\frac{2}{\alpha}}\frac{{\left(\prod_{i=1}^n t_i\right)}^{\alpha\phi-1}}{\left(\sum_{i=1}^n t_i^\alpha\right)^{n\phi}}  d\phi d\alpha \\
& > \int\limits_0^{\frac{2}{ne}}\int\limits_1^{\infty} c_3\alpha^{n-\frac{3}{2}}\phi^{\frac{n-1}{2}}\frac{{\left(\prod_{i=1}^n t_i\right)}^{\alpha\phi-1}}{\left(\sum_{i=1}^n t_i^\alpha\right)^{n\phi}}  d\phi d\alpha \\
&= \int\limits_0^{\frac{2}{ne}}c_3\alpha^{n-\frac{3}{2}}\frac{\Gamma\left(\frac{n+1}{2},n \f{p}(\alpha)\right)}{\left(n \f{p}(\alpha)\right)^{\frac{n+1}{2}}}d\alpha \geq \int\limits_0^{\frac{2}{ne}}g_3\alpha^{-\tfrac{5}{2}}d\alpha = \infty, 
\end{aligned}
\end{equation*}
where $c_3$ e $g_3$ are positive constants such that the above inequalities occur. \qedhere
\end{proof}

\subsection{Reference prior where $\alpha,\mu$ and $\phi$ are the ordered parameters}\label{sepalfao}
\vspace{0.3cm}

For the situation where the parameters of interest are ordered by $\boldsymbol\theta=(\alpha,\mu,\phi)$, through the Corollary 3.3 and using the Fisher information matrix (\ref{mfishergg}), we have
\begin{equation}\label{prioralfao1}
\pi(\alpha)\propto \frac{1}{\alpha}, \ \ \ \  \pi(\mu|\alpha,\phi)\propto\frac{1}{\mu}, \ \ \ \ \pi(\phi|\alpha,\mu)\propto\sqrt{\psi'(\phi)}.
\end{equation}

The joint prior distribution for the ordered parameters is given by,
\begin{equation}\label{prioralfao}
\pi_{\boldsymbol\theta}(\alpha,\mu,\phi)\propto \frac{\sqrt{\psi'(\phi)}}{\alpha\mu}.
\end{equation}

The joint posterior distribution for $\phi, \mu$ and $\alpha$, using the prior distribution (\ref{prioralfao}), is given by,
\begin{equation}\label{postalfao1}
p_{\boldsymbol\theta}(\phi,\mu,\alpha|\boldsymbol{t})=\frac{1}{d_4}\frac{\alpha^{n-1}\sqrt{\psi'(\phi)}}{\Gamma(\phi)^n}\mu^{n\alpha\phi-1}\left\{\prod_{i=1}^n{t_i^{\alpha\phi-1}}\right\}\exp\left\{-\mu^{\alpha}\sum_{i=1}^n t_i^\alpha\right\}. 
\end{equation}

\begin{proposition}  The posterior density (\ref{postalfao1}) is improper since, 
\begin{equation}\label{postalfao2}
d_4=\int\limits_{\mathcal{A}}\frac{\alpha^{n-1}\sqrt{\psi'(\phi)}}{\Gamma(\phi)^n}\mu^{n\alpha\phi-1}\left\{\prod_{i=1}^n{t_i^{\alpha\phi-1}}\right\}\exp\left\{-\mu^{\alpha}\sum_{i=1}^n t_i^\alpha\right\}d\boldsymbol{\theta}=\infty,
\end{equation}
where $\mathcal{A}=\{(0,\infty)\times(0,\infty)\times(0,\infty)\}$ is the parameter space for $\boldsymbol{\theta}$.
\end{proposition}
\begin{proof}
 Since $\frac{\alpha^{n-1}\sqrt{\psi'(\phi)}}{\Gamma(\phi)^n}\mu^{n\alpha\phi-1}\left\{\prod_{i=1}^n{t_i^{\alpha\phi-1}}\right\}\exp\left\{-\mu^{\alpha}\sum_{i=1}^n t_i^\alpha\right\}\geq 0$ by Tonelli theorem (see Folland, 1999) we have,
\begin{equation*}
\begin{aligned}
d_4&= \int\limits_0^{\infty}\int\limits_0^{\infty}\int\limits_0^{\infty}\frac{\alpha^{n-1}\sqrt{\psi'(\phi)}}{\Gamma(\phi)^n}\mu^{n\alpha\phi-1}\left\{\prod_{i=1}^n{t_i^{\alpha\phi-1}}\right\}\exp\left\{-\mu^{\alpha}\sum_{i=1}^n t_i^\alpha\right\}d\mu d\phi d\alpha \\
&=\int\limits_0^{\infty}\int\limits_0^{\infty}\alpha^{n-1}\frac{\sqrt{\psi'(\phi)}\Gamma(n\phi)}{\Gamma(\phi)^n}\frac{{\left(\prod_{i=1}^n t_i\right)}^{\alpha\phi-1}}{\left(\sum_{i=1}^n t_i^\alpha\right)^{n\phi}}d\phi d\alpha\\
& \geq \int\limits_0^{\infty}\int\limits_0^{1} c_1\alpha^{n-1}\phi^{n-2}\left(\frac{{\sqrt[n]{\prod_{i=1}^n t_i^\alpha}}}{\sum_{i=1}^n t_i^\alpha}\right)^{n\phi}d\phi d\alpha \\
&= \int\limits_0^{\infty}c_1\alpha^{n-1}\frac{\gamma\left(n-1,n \f{q}(\alpha)\right)}{\left(n \f{q}(\alpha)\right)^{n-1}}d\alpha \geq \int\limits_1^{\infty}g_1\frac{1}{\alpha} d\alpha = \infty, 
\end{aligned}
\end{equation*}
where $c_1$ and $g_1$ are positive constants such that the above inequalities occur. \qedhere
\end{proof}

\subsection{Modified Reference Prior}
\vspace{0.3cm}

Consider the case discuss in the Section \ref{sepalfao} where the parameters of interest are ordered by $\boldsymbol\theta$, a modification in $\pi(\alpha)$ was performed so that
\begin{equation}\label{prioralfamod1}
\pi(\alpha)\propto \frac{1}{\alpha^{-\tfrac{1}{2}+\tfrac{2\alpha}{1+\alpha}}}.
\end{equation}

This modification was made conveniently of such way that $\pi(\alpha)$ is proper, ie,
\begin{equation}\label{prioralfamod2}
\int_{0}^{\infty}\pi(\alpha)d\alpha=\int_{0}^{\infty} \frac{k^{-1}}{\alpha^{-\tfrac{1}{2}+\tfrac{2\alpha}{1+\alpha}}}d\alpha=1 \mbox{ e } k=\int_{0}^{\infty} \frac{1}{\alpha^{-\tfrac{1}{2}+\tfrac{2\alpha}{1+\alpha}}}d\alpha<\infty .
\end{equation}

The joint prior distribution for the ordered parameters is given by,
\begin{equation}\label{prioralfam}
\pi_M(\alpha,\mu,\phi)\propto \frac{\sqrt{\psi'(\phi)}}{\alpha^{-\tfrac{1}{2}+\tfrac{2\alpha}{1+\alpha}}\mu}.
\end{equation}

The joint posterior distribution for $\phi, \mu$ and $\alpha$, using the prior distribution (\ref{prioralfamod2}), is given by,
\begin{equation}\label{postalfam1}
p_R(\alpha,\phi,\mu|\boldsymbol{t})=\frac{1}{d_5}\frac{\alpha^{n+\tfrac{1}{2}-\tfrac{2\alpha}{1+\alpha}}\sqrt{\psi'(\phi)}}{\Gamma(\phi)^n}\mu^{n\alpha\phi-1}\left\{\prod_{i=1}^n{t_i^{\alpha\phi-1}}\right\}\exp\left\{-\mu^{\alpha}\sum_{i=1}^n t_i^\alpha\right\}. 
\end{equation}
\begin{theorem} The posterior density (\ref{postalfao1}) is proper since,
\begin{equation}\label{postalfamc1}
d_5=\int\limits_{\mathcal{A}}\frac{\alpha^{n+\tfrac{1}{2}-\tfrac{2\alpha}{1+\alpha}}\sqrt{\psi'(\phi)}}{\Gamma(\phi)^n}\mu^{n\alpha\phi-1}\left\{\prod_{i=1}^n{t_i^{\alpha\phi-1}}\right\}\exp\left\{-\mu^{\alpha}\sum_{i=1}^n t_i^\alpha\right\}d\boldsymbol{\theta}<\infty,
\end{equation}
where $\mathcal{A}=\{(0,\infty)\times(0,\infty)\times(0,\infty)\}$ is the parameter space for $\boldsymbol{\theta}$.
\end{theorem}
\begin{proof}
Since $\frac{\alpha^{n+\tfrac{1}{2}-\tfrac{2\alpha}{1+\alpha}}\sqrt{\psi'(\phi)}}{\Gamma(\phi)^n}\mu^{n\alpha\phi-1}\left\{\prod_{i=1}^n{t_i^{\alpha\phi-1}}\right\}\exp\left\{-\mu^{\alpha}\sum_{i=1}^n t_i^\alpha\right\}\geq0$ by Tonelli theorem (see Folland, 1999) we have,
\begin{equation*}\label{postdemapb2}
\begin{aligned}
d_5&= \int\limits_0^\infty \int\limits_0^\infty \int\limits_0^\infty \frac{\alpha^{n+\tfrac{1}{2}-\tfrac{2\alpha}{1+\alpha}}\sqrt{\psi'(\phi)}}{\Gamma(\phi)^n}\mu^{n\alpha\phi-1}\left\{\prod_{i=1}^n{t_i^{\alpha\phi-1}}\right\}\exp\left\{-\mu^{\alpha}\sum_{i=1}^n t_i^\alpha\right\}d\mu d\phi d\alpha \\
&=\int\limits_0^{\infty}\int\limits_0^{\infty}\frac{\alpha^{n-\tfrac{1}{2}-\tfrac{2\alpha}{1+\alpha}}\sqrt{\psi'(\phi)}}{\Gamma(\phi)^n}\left\{\prod_{i=1}^n{t_i^{\alpha\phi-1}}\right\}\dfrac{\Gamma(n\phi)}{\left(\sum_{i=1}^n t_i^\alpha\right)^{n\phi}}d\phi d\alpha=s_1+s_2+s_3+s_4,
\end{aligned}
\end{equation*}
where
\begin{equation*}
\begin{aligned}
s_1&=\int\limits_0^{1}\int\limits_0^{1}\frac{\alpha^{n-\tfrac{1}{2}-\tfrac{2\alpha}{(1+\alpha)}}\sqrt{\psi'(\phi)}}{\Gamma(\phi)^n}\left\{\prod_{i=1}^n{t_i^{\alpha\phi-1}}\right\}\dfrac{\Gamma(n\phi)}{\left(\sum_{i=1}^n t_i^\alpha\right)^{n\phi}}d\phi d\alpha \\
&< \int\limits_0^{1}c'_1\alpha^{n-\tfrac{1}{2}-\tfrac{2\alpha}{(1+\alpha)}}\frac{\gamma(n-1,n\f{q}(\alpha))}{(n\f{q}(\alpha))^{n-1}}d\alpha<\int\limits_0^{1} g'_1\alpha^{n-\tfrac{1}{2}} d\alpha < \infty,
\end{aligned}
\end{equation*}

\begin{equation*}
\begin{aligned}
s_2&=\int\limits_1^{\infty}\int\limits_0^{1}\frac{\alpha^{n-\tfrac{1}{2}-\tfrac{2\alpha}{(1+\alpha)}}\sqrt{\psi'(\phi)}}{\Gamma(\phi)^n}\left\{\prod_{i=1}^n{t_i^{\alpha\phi-1}}\right\}\dfrac{\Gamma(n\phi)}{\left(\sum_{i=1}^n t_i^\alpha\right)^{n\phi}}d\phi d\alpha\\&<\int\limits_1^{\infty}c'_1\alpha^{n-\tfrac{1}{2}-\tfrac{2\alpha}{(1+\alpha)}}\frac{\gamma(n-1,n\f{q}(\alpha))}{(n\f{q}(\alpha))^{n-1}}d\alpha<\int\limits_1^{\infty} g'_2\alpha^{-\tfrac{3}{2}} d\alpha < \infty,
\end{aligned}
\end{equation*}
\begin{equation*}
\begin{aligned}
s_3&=\int\limits_0^{1}\int\limits_1^{\infty}\frac{\alpha^{n-\tfrac{1}{2}-\tfrac{2\alpha}{(1+\alpha)}}\sqrt{\psi'(\phi)}}{\Gamma(\phi)^n}\left\{\prod_{i=1}^n{t_i^{\alpha\phi-1}}\right\}\dfrac{\Gamma(n\phi)}{\left(\sum_{i=1}^n t_i^\alpha\right)^{n\phi}}d\phi d\alpha, \\ &<\int\limits_0^1 c_2^{'}a^{n-\frac{1}{2}-\tfrac{2\alpha}{1+\alpha}}\frac{\Gamma(\frac{n-1}{2},n\f{p}(\alpha))}{(n\f{p}(\alpha))^{\frac{n}{2}}}d\alpha<\int\limits_0^{1} g'_3\alpha^{-\tfrac{1}{2}} d\alpha<\infty,
\end{aligned}
\end{equation*}
and
\begin{equation*}
\begin{aligned}
s_4&=\int\limits_1^{\infty}\int\limits_1^{\infty}\frac{\alpha^{n-\tfrac{1}{2}-\tfrac{2\alpha}{(1+\alpha)}}\sqrt{\psi'(\phi)}}{\Gamma(\phi)^n}\left\{\prod_{i=1}^n{t_i^{\alpha\phi-1}}\right\}\dfrac{\Gamma(n\phi)}{\left(\sum_{i=1}^n t_i^\alpha\right)^{n\phi}}d\phi d\alpha\\&<\int\limits_1^{\infty}c_2^{'}a^{n-\frac{1}{2}-\tfrac{2\alpha}{1+\alpha}}\frac{\Gamma(\frac{n-1}{2},n\f{p}(\alpha))}{(n\f{p}(\alpha))^{\frac{n}{2}}}d\alpha<\int\limits_1^{\infty} g'_4\alpha^{-2} d\alpha<\infty,
\end{aligned}
\end{equation*}
where $c'_1$, $c'_2$, $g'_1$, $g'_2$, $g'_3$ e $g'_4$ are positive constants such that the above inequalities occur. Therefore, we have: $d_5=s_1+s_2+s_3+s_4<\infty$. \qedhere
\end{proof}

The conditional posterior distributions for $\alpha,\mu$ and $\phi$ are given as
follows:
\begin{equation*}\label{postrejmodc1}
p_R(\alpha|\phi,\boldsymbol{t})\propto \frac{\alpha^{n+\tfrac{1}{2}-\tfrac{2\alpha}{1+\alpha}}}{\left(\sum_{i=1}^n t_i^\alpha\right)^{n\phi}}\left\{\prod_{i=1}^n{t_i^{\alpha\phi-1}}\right\},
\end{equation*}
\begin{equation}\label{postrejmodc2}
p_R(\phi|,\alpha,\boldsymbol{t})\propto \frac{\sqrt{\psi'(\phi)}\Gamma(n\phi)}{\Gamma(\phi)^n}\frac{\left\{\prod_{i=1}^n{t_i^{\alpha\phi-1}}\right\}}{\left(\sum_{i=1}^n t_i^\alpha\right)^{n\phi}},
\end{equation}
\begin{equation*}\label{postrejmodc3}
p_R(\mu|\phi,\alpha,\boldsymbol{t})\sim \f{GG}\left(n\phi,\left(\sum_{i=1}^n t_i^\alpha\right)^{1/\alpha},\alpha\right).
\end{equation*}

By considering the conditional posterior distributions (\ref{postrejmodc2}) the convergence of Monte Carlo Methods Markov Chain (MCMC) is easily achieved. Since the conditional distributions of $\alpha$ e $\phi$ have no closed form, the Metropolis-Hastings algorithm was used (see Gamerman and Lopes, 2006) to obtain the subsequent quantities.

\section{Simulation Study}
\vspace{0.3cm}

In this section we develop a simulation study via Monte Carlo methods. The main goal is to study the efficiency of the proposed method. The following procedure is adopted:
\begin{enumerate}
\item Generate values of the $\f{GG}(\phi,\mu,\alpha)$ with size $n$.
\item Using the values obtained in Step 1, calculate the posterior estimates $\hat{\phi},\hat{\mu}$ e $\hat{\alpha}$ using MCMC.
\item Repeat the Steps $2$ and $3$ $N$ times.
\item Using $\boldsymbol{\hat\theta}$ and $\boldsymbol{\theta}$, compute the mean relative estimates (MRE) $\sum_{i=1}^{N}\frac{\hat\theta/\theta}{N}$ and the mean square errors (MSE) $\sum_{i=1}^{N}\frac{(\hat\theta-\theta)^2}{N}$.
\end{enumerate}

It is expected that the MRE's are closer to one with smaller MSE.  

We fix the parameter values at $\boldsymbol\theta=(0.4,1.5,5)$, $N=1,000$ and $n=(50,100,150,200,250,300)$. It is also evaluated the coverage probability of the $95\%$ highest posterior density ($HPD_{95\%}$) intervals (see Migon, Gamerman and Louzada, 2014), that is, the Lower (L), Upper (U) and $95\%$ (C) amounts of coverage. 
For a large number of experiments, using a confidence level of $95\%$, the frequencies of intervals that covered the true values of $\boldsymbol{\theta}$ should be close $95\%$.
 
For each simulated sample, $31,000$ iterations were performed using MCMC methods. As a sample burning, we discard the first $1,000$ initial values. To reduce the autocorrelation values of the chains, we saved generated samples with space of size $30$, getting at the end three chains of size $1,000$. To check the convergence of the obtained chains, we used the Geweke criterion (Geweke, 1992) with a confidence level of $95\%$. We computed the posterior mode estimate, yielding $2,000$ estimates for $\phi,\mu$ and $\alpha$ using different simulated samples.

Tables 1 shows the MRE's, MSE's, the coverage probability L, U and C with a confidence level equals to $95\%$ from the estimates obtained of the posterior modes for $1,000$ simulated samples, considering different values of $n$.
It can be observed that the MSE's decreases as $n$ increases and also, as expected, the values of MRE's tend to 1, allowing to obtain good inferences for the parameters of the GG distribution. It is important to point out that using this approach, the coverage probability (C) for all parameters tend to $0.95$ and the values of (L) and (U) are approaching $0.025$ as there is an increase of the size of $n$. 

\begin{table}[ht]
\centering
\caption{MRE, MSE, L, U and C estimates for $1000$ samples of sizes $n=(50,100,200,250,300)$.}
\begin{tabular}{c r r r r r r }
  \hline
 &  \multicolumn{1}{c}{$n$} & \multicolumn{1}{c}{MRE} & \multicolumn{1}{c}{MSE} & \multicolumn{1}{c}{L} & \multicolumn{1}{c}{U}& \multicolumn{1}{c}{C}  \\ 
	\hline
$\phi=0.4$ &  \multicolumn{1}{c}{$50$}      & 0.668 & 0.052 & 0.010 & 0.000 & 0.990  \\
&  \multicolumn{1}{c}{$100$}     & 0.789 & 0.035 & 0.018 & 0.056 & 0.926  \\
& \multicolumn{1}{c}{$150$}     & 0.862 & 0.019 & 0.015 & 0.037 & 0.948  \\
&	\multicolumn{1}{c}{$200$}     & 0.898 & 0.016 & 0.017 & 0.029 & 0.954  \\
&  \multicolumn{1}{c}{$250$}     & 0.919 & 0.012 & 0.013 & 0.046 & 0.941  \\
&	\multicolumn{1}{c}{$300$}     & 0.942 & 0.010 & 0.019 & 0.026 & 0.955  \\

	\hline
$\mu=1.5$ & \multicolumn{1}{c}{$50$}     & 0.950 & 0.033 & 0.016 & 0.001 & 0.983 \\
&  \multicolumn{1}{c}{$100$}      & 0.967 & 0.018 & 0.016 & 0.058 & 0.926 \\
&  \multicolumn{1}{c}{$150$}     & 0.979 & 0.010 & 0.014 & 0.035 & 0.951 \\
&	\multicolumn{1}{c}{$200$}      & 0.984 & 0.008 & 0.016 & 0.032 & 0.952 \\
&  \multicolumn{1}{c}{$250$}      & 0.988 & 0.007 & 0.019 & 0.038 & 0.943 \\
&	\multicolumn{1}{c}{$300$}     & 0.991 & 0.006 & 0.021 & 0.023 & 0.956 \\
	\hline

$\alpha=5.0$ & \multicolumn{1}{c}{$50$}      & 0.952 & 3.398 & 0.006 & 0.012 & 0.982 \\
&  \multicolumn{1}{c}{$100$}     & 1.083 & 4.551 & 0.055 & 0.017 & 0.928 \\
&  \multicolumn{1}{c}{$150$}     & 1.042 & 2.086 & 0.038 & 0.013 & 0.949 \\
&	\multicolumn{1}{c}{$200$}     & 1.028 & 1.470 & 0.032 & 0.013 & 0.955 \\
&	\multicolumn{1}{c}{$250$}     & 1.022 & 1.183 & 0.045 & 0.016 & 0.939 \\
&	\multicolumn{1}{c}{$300$}     & 1.010 & 0.915 & 0.033 & 0.023 & 0.944 \\
	\hline
 
\end{tabular}
\end{table}

\section{Industrial experimental data}
\vspace{0.3cm}

In this section, it is considered a dataset related to the lifetimes of 30 units in an industrial experiment (Table 2) extracted from Meeker \& Escobar (1998). 

\begin{table}[ht]
\caption{Dataset related to the lifetimes of $30$ units in an industrial experiment (Meeker \& Escobar, 1998, p.383).}
\centering  
\begin{tabular}{c c c c c c c c c c }  
\hline  
275 & 13 & 147 & 23 & 181 & 30 & 65 & 10 & 300 & 173 \\
106 & 300 & 300 & 212 & 300 & 300 & 300 & 2 & 261 & 293 \\
88 & 274 & 28 & 143 & 300 & 23 & 300 & 80 & 245 & 266 \\ [0ex] 
\hline  
\end{tabular}
\end{table}

We consider the GG distribution to analyze this dataset and compare the results with the fits of some of its particular cases: Weibull, Log Normal and Gamma distribution through the BIC (Schwarz, 1978) and DIC tests (Spiegelhalter et al., 2002).
Figures containing the generated samples of the marginal posterior densities for $\phi,\mu$ and $\alpha$, through MCMC methods, time-series plots of the simulated samples and the sample autocorrelations are available in Figure \ref{f-apli-dsa-1}.

It is observed from Figure \ref{f-apli-dsa-1} and also through the Geweke criterion, that there is an indication of convergence of the simulation MCMC algorithm. The results of the simulated chains can be assumed to be samples of the marginal posterior distributions for $\phi,\mu$ and $\alpha$. Due to the asymmetry of the marginal posterior distributions, we estimated the posterior modes of the posterior distribution for $\boldsymbol{\theta}$ (see Table 3).
\begin{table}[!ht]
\caption{Posterior modes, posterior standard deviations, and $95\%$ credibility intervals for the parameters $\phi,\mu$ and $\alpha$.}
\centering  
\begin{center}
  \begin{tabular}{ c  c   c c}
    \hline
		$\boldsymbol{\theta}$  & Mode & SD & $HPD_{95\%}(\theta)$ \\ \hline
    \ \ $\phi$ \ \   & 0.08135 & 0.06149 &  (0.04560; 0.28829) \\ \hline
    \ \ $\mu$   \ \  & 0.00299 & 0.00024 &  (0.00235; 0.00327) \\ \hline
    \ \ $\alpha$ \ \ & 10.48558 & 4.51768 &  (3.38645; 21.19713) \\ \hline
  \end{tabular}
\end{center}
\end{table}

\begin{figure}[!htb]
\centering
\includegraphics[scale=0.55]{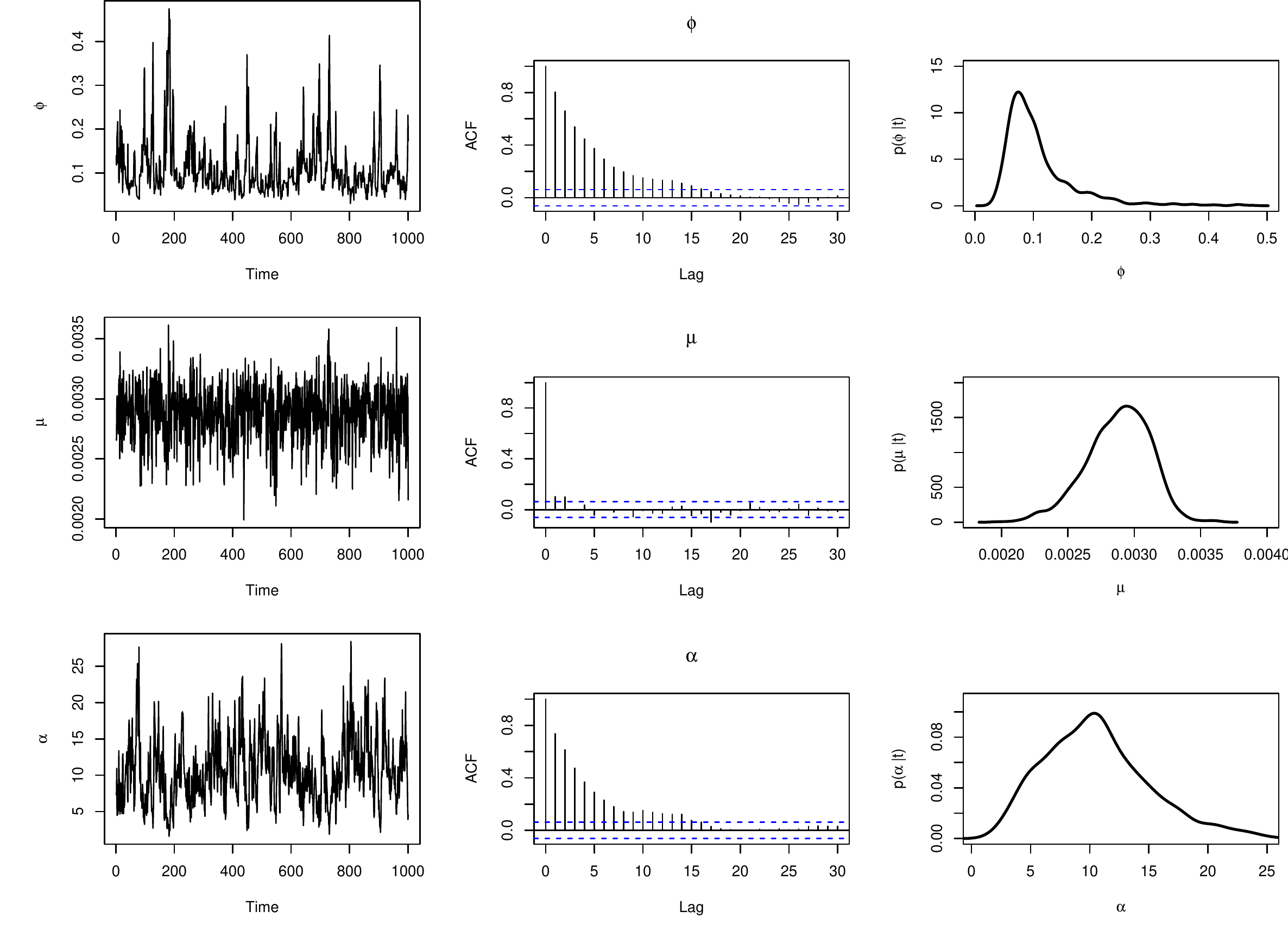}
\caption{Plots of temporal autocorrelations and the marginal posterior densities for the generated series for $\phi,\mu$ and $\alpha$}.\label{f-apli-dsa-1}
\end{figure}

Table 4 shows the results of DIC and BIC criteria for different probability distributions, considering the presented data.

\begin{table}[!ht]
\caption{Results of DIC and BIC criteria for different probability distributions.}
\centering  
\begin{center}
  \begin{tabular}{ c  c  c  c c}
    \hline
		Criterion  & Generalized Gamma & \ Weibull \ & \ \ Gamma \ \ & Log-normal \\ \hline
      DIC & \textbf{197.517} & 363,60 & 374,50 & 386,10 \\ \hline
      BIC & \textbf{362.765} & 385,70 & 377,18 & 389,08 \\ \hline
  \end{tabular}
\end{center}
\end{table}

Given a set of fitted candidate distributions for the lifetime $T$, the preferred one is given by the distribution which provides the lowest BIC and DIC. Based on both criteria, we can conclude from the results from the Table 4 that the GG distribution is the best distribution for the industrial data set.

\newpage
\section{Final Comments}
\vspace{0.3cm}

In this paper, different reference prior distributions for the GG distribution were obtained. We proved that such priors lead to improper posterior distributions and should not be used in a Bayesian framework. Such a problem is overcome via a modification made in the reference prior distribution, where the parameters of interest are ordered.  
It was proved that using such a prior we obtain a proper posterior distribution. 

From the practical point of view, based on a simulation study and on an important data from Meeker \& Escobar (1998), we demonstrated that using the obtained posterior distribution it is possible to obtain good estimates of the parameters of the GG distribution.  

\section*{Acknowledgements}

The research was partially supported by CNPq, FAPESP and CAPES of Brazil.

\bibliographystyle{abbrv}

\end{document}